\documentclass[a4paper]{article}
\usepackage[utf8]{inputenc}
\usepackage{lmodern} 
\usepackage{tikz}
\usepackage{amsmath}
\usepackage[hyphens]{url}
\usepackage{hyperref}
\usepackage[hyphenbreaks]{breakurl}
\usepackage{xurl}
\hypersetup{breaklinks=true}
\usepackage{cleveref}
\usepackage{amsthm}
\usepackage{orcidlink}
\usepackage{mathtools}
\usepackage[ruled,vlined,linesnumbered]{algorithm2e}
\crefname{algocf}{alg.}{algs.}
\Crefname{algocf}{Algorithm}{Algorithms}
\crefname{problem}{problem}{problems}
\Crefname{problem}{Problem}{Problems}
\usepackage{booktabs} 

\DeclareMathOperator*{\argmax}{argmax}
\DeclareMathOperator*{\leftext}{lext}
\DeclareMathOperator*{\rightext}{rext}

\newcommand{\forward}{\mathtt{forward}}

\newcommand{\No}[1]{}

\newtheorem{theorem}{Theorem}
\crefname{theorem}{theorem}{theorems}
\Crefname{theorem}{Theorem}{Theorems}

\newtheorem{lemma}{Lemma}
\crefname{lemma}{lemma}{lemmas}
\Crefname{lemma}{Lemma}{Lemmas}

\newtheorem{corollary}{Corollary}
\crefname{corollary}{corollary}{corollaries}
\Crefname{corollary}{Corollary}{Corollaries}

\newtheorem{definition}{Definition}
\crefname{definition}{definition}{definitions}
\Crefname{definition}{Definition}{Definitions}

\newtheorem{problem}{Problem}
\crefname{problem}{problem}{problems}
\Crefname{problem}{Problem}{Problems}

\crefname{observation}{observation}{observations}
\Crefname{observation}{Observation}{Observations}

\crefname{remark}{remark}{remarks}
\Crefname{remark}{Remark}{Remarks}

\crefname{assumption}{assumption}{assumptions}
\Crefname{assumption}{Assumption}{Assumptions}

\newtheorem{innercustomthm}{Theorem}
\newenvironment{customthm}[1]
  {\renewcommand\theinnercustomthm{#1}\innercustomthm}
  {\endinnercustomthm}

\newtheorem{innercustomcrl}{Corollary}
\newenvironment{customcrl}[1]
  {\renewcommand\theinnercustomcrl{#1}\innercustomcrl}
  {\endinnercustomcrl}

\usepackage{xcolor} 

\def\HiLi{\leavevmode\rlap{\hbox to \hsize{\color{gray!25}\leaders\hrule height .8\baselineskip depth .5ex\hfill}}}
\def\HiLiHybrid{\leavevmode\rlap{\hbox to \hsize{\color{gray!10}\leaders\hrule height .8\baselineskip depth .5ex\hfill}}}

\begin{document}

\title{Chaining of Maximal Exact Matches in Graphs}

\date{
Department of Computer Science, University of Helsinki, Finland
}

\author{
Nicola Rizzo\,\orcidlink{0000-0002-2035-6309}\\\texttt{nicola.rizzo@helsinki.fi} \and Manuel C\'aceres\,\orcidlink{0000-0003-0235-6951}\\\texttt{manuel.caceres@helsinki.fi} \and Veli M\"akinen\,\orcidlink{0000-0003-4454-1493}\\\texttt{veli.makinen@helsinki.fi}
}

\maketitle

\begin{abstract}
We show how to chain \emph{maximal exact matches} (MEMs) between a query string $Q$ and a labeled directed acyclic graph (DAG) $G=(V,E)$ to solve the \emph{longest common subsequence} (LCS) problem between $Q$ and $G$. We obtain our result via a new symmetric formulation of chaining in DAGs that we solve in $O(m+n+k^2|V| + |E| + kN\log N)$ time, where $m=|Q|$, $n$ is the total length of node labels, $k$ is the minimum number of paths covering the nodes of $G$ and $N$ is the number of MEMs between $Q$ and node labels, which we show encode full MEMs.
\\[1ex]\textbf{Keywords} sequence to graph alignment, longest common subsequence, sparse dynamic programming
\\[1ex]\textbf{Funding} This project received funding from the European Union’s Horizon 2020 research and innovation programme under the Marie Skłodowska-Curie grant agreement No.\ 956229, and from the Academy of Finland grants No.\ 352821 and 328877.
\end{abstract}

\section{Introduction}

Due to recent developments in \emph{pangenomics}~\cite{Maretal16} there is a high interest to extend the notion of string alignments to graphs. A common pangenome representation is a node-labeled directed acyclic graph (DAG), whose paths represent plausible individual genomes from a species. 
Unfortunately, even finding an exact occurrence of a query string as a subpath in a graph is a conditionally hard problem~\cite{equi2023complexity,EMT21}: only quadratic time dynamic programming solutions are known and faster algorithms would contradict the Strong Exponential Time Hypothesis (SETH). Due to this theoretical barrier, parameterized solutions have been developed ~\cite{DBLP:conf/cpm/Caceres23,cotumaccio2021indexing,DBLP:conf/dcc/Cotumaccio22,rizzo2022solving}, and/or the task has been separated into finding short exact occurrences (anchors) and then \emph{chaining} them into longer matches~\cite{Maketal19,li2020design,Ma2022.01.07.475257,GC23}. Although the chaining algorithms provide exact solutions to their internal chaining formulations and their solutions can be interpreted as alignments of queries to a graph with edit operations, so far they have not been shown to provide exact solutions to the corresponding alignment formulation.

In this paper, we integrate a symmetric formulation from string chaining \cite{SK03,MS20} to graph chaining  \cite{Maketal19} yielding the first chaining-based parameterized exact alignment algorithm between a query string and a graph. Namely, we obtain an $O(m+n+k^2|V|+|E|+kN\log N)$ time algorithm for computing the length of a \emph{longest common subsequence} (LCS) between a query string $Q$ and a path of $G$, where $m=|Q|$, $n$ is the total length of node labels, $k$ is the width (minimum number of paths covering the nodes) of $G$, and $N$ is the number of \emph{maximal exact matches} (MEMs) between $Q$ and the node labels (node MEMs).

The paper is structured as follows. The preliminaries in \Cref{sect:preliminaries} and the basic concepts in \Cref{sect:nodeMEMs} follow the notions developed in our recent work~\cite{RCM23b}, where we introduce the definition of a MEM between a string and a graph, and study the non-trivial problem of finding graph MEMs with a length threshold; for the purposes of this paper, we observe that node MEMs are sufficient.
In \Cref{sect:DAGCLCrevisited}, we revise the solution for an asymmetric chaining formulation in DAGs~\cite{Maketal19} for the case of node MEMs. Then, in \Cref{sect:symmetricCLCrevisited}, we tailor the string to string symmetric chaining algorithm \cite{SK03,MS20} to use MEM anchors. In \Cref{sect:symmetricDAGCLCintegration}, we show how to integrate these two approaches to obtain our main result. Finally, in \Cref{sect:discussion} we discuss the length threshold setting and cyclic graphs.

\section{Preliminaries}\label{sect:preliminaries}
\paragraph*{Strings.}
We work with strings coming from a finite alphabet $\Sigma = [1..\sigma]$ and assume that $\sigma$ is at most the length of the strings we work with. For two integers $x$ and $y$ we use $[x..y]$ to denote the integer interval $\{x, x+1, \ldots, y\}$ or the empty set $\emptyset$ when $x>y$. 
A \emph{string} $T$ is an element of $\Sigma^n$ for a non-negative integer $n$, that is  sequence of $n$ symbols from $\Sigma$, where $n = |T|$ is the \emph{length} of the string. We denote $\varepsilon$ to the only string of length zero. We also denote $\Sigma^+ = \Sigma^*\setminus \{\varepsilon\}$.
For two strings $T_1$ and $T_2$ we denote their \emph{concatenation} as $T_1 \cdot T_2$, or just $T_1 T_2$. For a set of integers $I$ and a string $T$, we use $T[I]$ to denote the \emph{subsequence} of $T$ made of the concatenation of the characters indicated by $I$ in increasing order. If $I$ is an integer interval $[x..y]$, then $T[x..y]$ is a \emph{substring}: if $x = y$ then we also use $T[x]$, if $y<x$ then $T[x..y] = \varepsilon$, if $x\le y = n$ we call it a \emph{suffix} (\emph{proper suffix} when $x > 1$) and if $1 = x \le y$ we call it a \emph{prefix} (\emph{proper prefix} when $y < n$). 
A length-$\kappa'$ substring $Q[x..x+\kappa'-1]$ \emph{occurs} in $T$ if $Q[x..x+\kappa'-1] = T[i..i+\kappa'-1]$; in this case, we say that $(x,i,\kappa')$ is an \emph{(exact) match} between $Q$ and $T$, and \emph{maximal} (a MEM) if the match cannot be extended to the left (\emph{left-maximality}), that is, $x_1 = 1$ or $x_2 = 1$ or $Q[x_1-1] \neq T[x_2-1]$ nor it can be extended to the right (\emph{right-maximality}) $x_1 + \ell = \lvert Q \rvert$ or $x_2 + \ell = \lvert T \rvert$ or $Q[x_1 + \ell] \neq T[x_2 + \ell]$.

\paragraph*{Labeled graphs.}

We work with labeled directed acyclic graphs (DAGs) $G=(V,E,\ell)$, where $V$ is the vertex set, $E$ the edge set, and $\ell: V \to \Sigma^+$ a \emph{labeling} function on the vertices. A length-$k$ \emph{path} $P$ from $v_1$ to $v_k$ is a sequence of nodes $v_1, \ldots, v_k$ such that $(v_1,v_2),(v_2,v_3),\ldots,(v_{k-1},v_k) \in E$, in this case we say that $v_1$ \emph{reaches} $v_k$. We extend the labeled function to paths by concatenating the corresponding node labels, that is, $\ell(P) \coloneqq \ell(v_1) \cdots \ell(v_k)$. 
For a node $v$ and a path $P$ we use $\lVert\cdot\rVert $ to denote its \emph{string length}, that is $\lVert v\rVert  = |\ell(v)|$ and $\lVert P\rVert  = |\ell(P)|$. We say that a length-$\kappa'$ substring $Q[x..x+\kappa'-1]$ \emph{occurs} in $G$ if $Q[x..x+\kappa'-1]$ occurs in $\ell(P)$ for some path $P$. In this case, we say that $([x..x+\kappa'-1], (i, P = v_1\ldots v_k, j))$ is an \emph{(exact) match} between $Q$ and $G$, where $Q[x..x+\kappa'-1] = \ell(v_1)[i..] \cdot \ell(v_2) \cdots \ell(v_{k-1}) \cdot \ell(v_k)[..j]$, with $1 \le i \le \lVert v_1\rVert $ and $1 \le j \le \lVert v_k\rVert $. We call the triple $(i, P, j)$ a \emph{substring} of $G$ and we define its \emph{left-extension} $\leftext(i,P,j)$ as the singleton $\lbrace \ell(v_1)[i-1] \rbrace$ if $i > 1$ and $\lbrace \ell(u)[\lVert u\rVert ] \mid (u,v_1) \in E \rbrace$ otherwise.
Analogously, the \emph{right-extension} $\rightext(i,P,j)$ is $\lbrace \ell(v_k)[j+1] \rbrace$ if $j < \lVert v_k\rVert $ and $\lbrace \ell(v)[1] \mid (v_k, v) \in E \rbrace$ otherwise.
Note that the left (right) extension can be equal to the empty set $\emptyset$, if the start (end) node of $P$ does not have incoming (outgoing) edges. See \Cref{fig:subpath}. 

\paragraph*{Chaining of matches.}

An \emph{asymmetric chain} $A'[1..N']$ is an ordered subset of a set $A$ of $N$ exact matches between a labeled DAG $G=(V,E,\ell)$ and a query string $Q$,
with the ordering $A'[l]<A'[l+1]$ for $1\leq l<N'$ defined as $([x'..x'+\kappa''-1], (i',P_l,j'))<([x..x+\kappa'-1], (i,P_{l+1},j))$ iff the start of path $P_{l+1}$ is strictly reachable from the end of path $P_l$ and $x'\leq x$. The asymmetry comes from the fact that overlaps are not allowed in $G$, but they are allowed in $Q$. We are interested in chains that maximize the length of an induced subsequence $Q'$, denoted $Q'=Q \mid A'$, that is obtained by deleting all parts of $Q$ that are not covered by chain $A'$. For example, consider $Q=\mathtt{ACATTCAGTA}$ and $A'=([2..4], (i_1,P_1,j_1)),([3..6], (i_2,P_2,j_2)), ([9..10], (i_3,P_3,j_3))$. Then $Q' = Q \mid A' =\mathtt{CATTCTA}$; anchors cover the underlined part of $Q=\mathtt{A\underline{CATTC}AG\underline{TA}}$.

We could define symmetric chains by considering overlaps of paths, but for the purposes of this paper it will be sufficient to consider overlaps of exact matches inside the nodes of $G$: A \emph{symmetric chain} $A'[1..N']$ is an ordered subset of a set $A$ of $N$ exact matches between the nodes of a labeled DAG $G=(V,E,\ell)$ and a query string $Q$,
with the ordering $A'[l]<A'[l+1]$ for $1\leq l<N'$ defined as $([x'..x'+\kappa''-1], (i',v,j'))<([x..x+\kappa'-1], (i,w,j))$ iff (i) $w$ is strictly reachable from $v$ or $v=w$ and $i'\leq i$ and (ii) $x'\leq x$. We extend the notation $Q'=Q \mid A'$ to cover symmetric chains $A'$ so that $Q'$ is obtained by deleting all parts of $Q$ that are not \emph{mutually} covered by chain $A'$. We define mutual coverage in \Cref{sect:symmetricCLCrevisited}: Informally, $Q'$ is formed by concatenating the prefixes of exact matches until reaching the overlap between the next exact match in the chain. \Cref{fig:subpath} illustrates the concept.

\begin{figure}
\centering
\includegraphics[width=.55\textwidth]{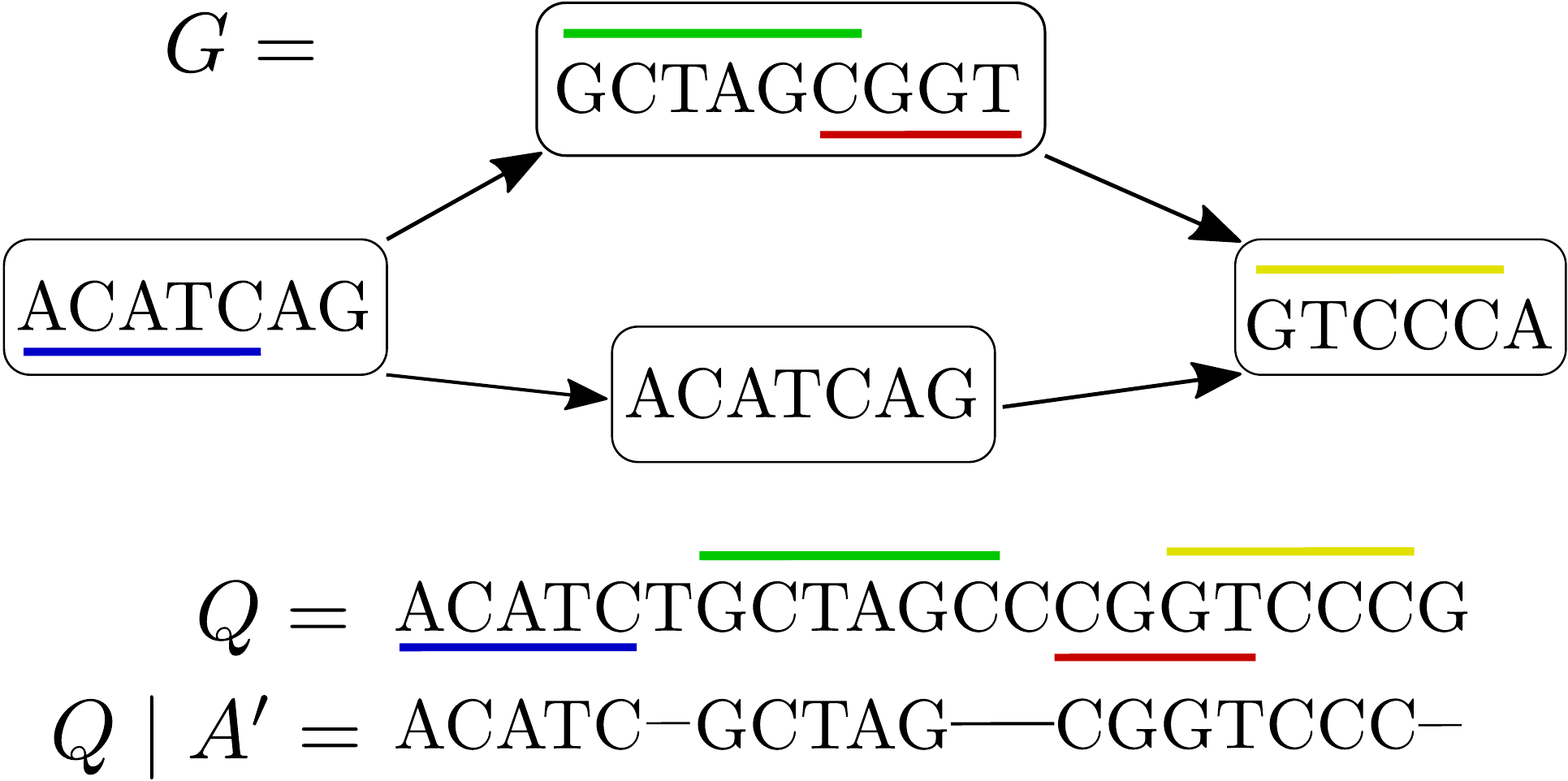}
\caption{Co-linear chaining setting between a string $Q$ and a labeled graph $G$. If $v$ is the last node to the right, then $([16..20],(1,v,5))$ is a match, with $\leftext(1,v,5) = \lbrace \mathtt{G}, \mathtt{T} \rbrace$ and $\rightext(1,v,5) = \lbrace \mathtt{A} \rbrace$. It is a MEM, since $\lvert \leftext(i,P,j) \rvert \ge 2$ and it cannot be extended to the right (\Cref{def:mem}). In fact, all exact matches are MEMs and they form a symmetric chain $A'$ (blue-green-red-yellow) inducing the subsequence $Q \mid A'$ (the last $\mathtt{C}$ of green match is omitted due to overlap with the red match).}\label{fig:subpath}
\end{figure}

\section{Finding MEMs in labeled DAGs\label{sect:nodeMEMs}}

We now consider the problem of finding all \emph{maximal exact matches} (MEMs) between a labeled graph $G$ and a query string $Q$ for the purpose of chaining.
\begin{definition}[MEM between a pattern and a graph~\cite{RCM23b}]\label{def:mem}
Let $G = (V,E,\ell)$ be a labeled graph, with $\ell \colon V \to \Sigma^+$, and $Q \in \Sigma^+$. We say that a match $([x..y],(i,P,j))$ between $Q$ and $G$ is \emph{left-maximal} (\emph{right-maximal}) if it cannot be extended to the left (right) in both $Q$ and $G$, that is,
\begin{align*}
    (\mathsf{LeftMax}) & \quad x=1 \vee \leftext(i,P,j) = \emptyset \vee Q[x-1] \notin \leftext(i,P,j) \qquad\text{and}
     \\
    (\mathsf{RightMax}) & \quad y = \lvert Q \rvert \vee \rightext(i,P,j) = \emptyset \vee Q[y+1] \notin \rightext(i,P,j).
\end{align*}
The pair $([x..y],(i,P,j))$ is a MEM if it is left-maximal or its left (graph) extension is not a singleton, and right-maximal or its right (graph) extension is not a singleton, that is
    $\mathsf{LeftMax} \vee \lvert \leftext(i,P,j) \rvert\geq 2$
    and
    $\mathsf{RightMax} \vee \lvert \rightext(i,P,j) \rvert\geq 2$.
\end{definition}

See \Cref{fig:subpath} for an example.
We use this particular extension of MEMs to graphs---with the additional conditions on non-singletons $\leftext$ and $\rightext$---as it captures all MEMs between $Q$ and $\ell(P)$, where $P$ is a source-to-sink path in $G$.
Moreover, we will show that this MEM formulation captures LCS through co-linear chaining, whereas avoiding the additional conditions would fail. Indeed, consider $Q$, $G$, and match $([16..20],(1,v,5))$ from \Cref{fig:subpath}: the match is not left-maximal, since $Q[15] = \mathtt{G}$ and $\mathtt{G} \in \leftext(1,v,5)$, but extending it would impose any chain using it as an anchor to go through the bottom suboptimal path, that in this case does not capture the LCS between $Q$ and $G$. Also, it turns out that we can focus on MEMs between the node labels and the query, as chaining will cover longer MEMs implicitly. 

To formalize the intuition, we say that a \emph{node MEM} is a match $(i,P,j)$ of $Q[x..y]$ in $G$ such that $P = v$ and it is left and right maximal w.r.t. $\ell(P)$ only in the string sense: conditions $\mathsf{LextMax} \vee i = 1$ and $\mathsf{RightMax} \vee j = \lVert v \rVert$ hold. 
Consider the text
\[
    T_\text{nodes} = \prod_{v \in V} \mathbf{0} \cdot \ell(v),
\]
where $\mathbf{0}\notin \Sigma$ is used as a delimiter to prevent MEMs spanning more than a node label. Running the MEM finding algorithm~\cite{BCKM20} on $Q$ and $T_\text{nodes}$ will retrieve exactly the node MEMs we are looking for \cite{RCM23b} (a more involved problem of finding graph MEMs with a length threshold is studied in \cite{RCM23b}, but here a simplified result without the threshold is sufficient): 

\begin{lemma}[\cite{RCM23b}]\label{lemma:nodeMEMs}
    Given a labeled DAG $G = (V,E,\ell)$, with $\ell: V \to \Sigma^+$, and a query string $Q$, we can compute all node MEMs between $Q$ and $G$ in time $O(n+m+N)$, where $n$ is the total length of node labels, $m = |Q|$, and $N$ is the number of node MEMs.
\end{lemma}

Let $A$ be the set of node MEMs found using \Cref{lemma:nodeMEMs}.
In \Cref{app:longMEMs}, we show that any long MEM spanning two or more nodes in $G$ can be formed by concatenating node MEMs into \emph{perfect chains}---chains that have no gap between consecutive matches.

\begin{theorem}
\label{theorem:longMEMs}
For every MEM $([x..y], (i,P,j))$ between $G$ and $Q$, there is a perfect chain $A'[1..p]\subseteq A$ such that $A'[1] \cdots A'[p] = ([x..y], (i,P,j))$.
\end{theorem}

\begin{corollary}\label{cor:nodemems}
The set $A$ is a \emph{compact representation} of the set $M$ of MEMs between query $Q$ and a labeled DAG $G=(V,E,\ell)$: it holds $|A|\leq \lVert M\rVert $, where $\lVert M\rVert $ is the length of the encoding of the paths in MEMs as the explicit sequence of its nodes. 
\end{corollary}

Our strategy is to use set $A$ as the representation of MEMs: Perfect chains are implicitly covered by the chaining algorithms of next section.

\section{Symmetric co-linear chaining in labeled DAGs}
\label{sect:clc}
M\"akinen et al. \cite[Theorem 6.4]{Maketal19} gave an $O(kN \log N + k|V|)$-time algorithm to find an asymmetric chain $A'[1..N']$ of a set $A$ of $N$ anchors\footnote{Anchors have the same representation as graph MEMs, $([x..y], (i,P,j))$, but they do not necessarily represent exact matches.} between a labeled DAG $G=(V,E,\ell)$ and a query string $Q$ maximizing the length of an induced subsequence $Q'=Q \mid A'$. Here $k$ is the \emph{width} of $G$, that is, the minimum number of paths covering nodes $V$ of $G$. The algorithm assumes a minimum path cover as its input, which can be computed in $O(k^2|V|+|E|)$ time \cite{Arieletal22,caceres2022minimum}. A limitation of this chaining algorithm is that anchors in the solution are not allowed to overlap in the graph, which has been partially solved by considering one-node overlaps~\cite{Ma2022.01.07.475257}. However, both of these approaches maximize the length of the sequence induced by the reported chain only on the string $Q$, which makes the problem formulation asymmetric.

In the case of two strings as input, the asymmetry of the coverage metric was solved by M\"akinen and Sahlin \cite{MS20} applying the technique by Shibuya and Kurochkin \cite{SK03}. They provided an $O(N \log N)$-time algorithm to find a symmetric chain $A'[1..N']$ of a set $A$ of $N$ anchors maximizing the length of an \emph{induced common subsequence} $C=Q \mid A'=T \mid A'$ between two input strings $Q$ and $T$, that is obtained by deleting all parts of $Q$, or equivalently all parts of $T$, that are not \emph{mutually covered} by chain $A'$ (to be defined below). 
Here anchors are assumed to be exact matches $(x,i,\kappa')$ (not necessarily maximal) such that $Q[x..x+\kappa'-1]=T[i..i+\kappa'-1]$, and $A'[j]<A'[j+1]$ for $1\leq j<N'$, where the order $<$ between anchors is defined as  $(x',i',\kappa'')<(x,i,\kappa')$ iff $x'\leq x$ and $i'\leq i$.
For completeness, in \Cref{app:proof-lcs-string} we include a revised proof that this algorithm computes the length of a longest common subsequence of strings $Q$ and $T$ if it is given all (string) MEMs between $Q$ and $T$ as input~\cite{MS20}.
The concept of mutual coverage~\cite[Problem 1]{MS20} is defined through the score
\[
\mathtt{coverage}(A')=\sum_{j=1}^{N'} \min_{\begin{array}{c} (i,x,\kappa'):=A'[j+1],\\(i',x',\kappa''):=A'[j] \end{array}} \left\{\begin{array}{l} \min(i,i'+\kappa'')-i', \\ 
\min(x,x'+\kappa'')-x',
\end{array}\right.
\]
where $A'[N'+1]=(\infty,\infty,0)$.
Each part of the sum contributes the corresponding number of character matches from the beginning of the anchors to the induced common subsequence. These form the mutually covered part of the inputs; see \Cref{fig:subpath} for an illustration on an extension of this concept to graphs.

Consider now the symmetric chaining problem between a DAG and a string: 

\begin{problem}[Symmetric DAG chaining with overlaps]
Find a symmetric chain $A'[1..N']$ of a set $A$ of $N$ anchors between a labeled DAG $G=(V,E,\ell)$ and a query string $Q$ maximizing the length of an induced common subsequence $C=P \mid A'=Q \mid A'$ for some path $P$ of $G$, where $P \mid A'$ denotes the subsequence obtained by deleting the parts of $\ell(P)$ that are not mutually covered by chain $A'$ and $Q \mid A'$ denotes the subsequence obtained by deleting the parts of $Q$ not mutually covered by chain $A'$.
\label{prob:symmetricdagchaining}
\end{problem}

In this section, we will solve this problem in the special case where the anchors are all node MEMs between $G$ and $Q$: thanks to \Cref{theorem:longMEMs} we know that the algorithm by M\"akinen et al. \cite{Maketal19} solves the problem when a longest induced common subsequence $C$ is covered by node MEMs that appear in different nodes. Since in our setting the overlaps can only occur inside node labels, we are left with what essentially is the symmetric string-to-string chaining problem \cite{SK03,MS20}. However, we cannot separate these subproblems and call the respective algorithms as black boxes, but instead we need to carefully interleave the computation of both techniques in one algorithm.

\subsection{DAG chaining with node MEMs\label{sect:DAGCLCrevisited}} 

\Cref{algo:colinearchainingDAGkpathcover} shows the pseudocode of~\cite[Algorithm 1]{Maketal19} simplified to take node MEMs as anchors. The original algorithm uses two arrays to store the start and the end nodes of anchor paths, but in the case of node MEMs one array suffices. We also modified~\cite[Lemma 3.2]{Maketal19} below to explicitly use primary and secondary keys (the original algorithms~\cite{Maketal19,MS20} implicitly assumed distinct keys). We still use primary keys to store MEM ending positions in $Q$ to do range searches, and we use the secondary key to store the MEM identifiers to update the values of the corresponding anchors.

\begin{algorithm}[htp]
\KwIn{A DAG $G=(V,E,\ell)$, a query string $Q$, a path cover $P_1,P_2,\ldots, P_k$ of $G$, and node MEMs $A[1..N]$ of the form $([x..x+\kappa'-1], (i,v,i+\kappa'-1))$, where $\ell(v)[i..i+\kappa'-1]=Q[x..x+\kappa'-1].$}
\KwOut{Index of a MEM ending at a chain with maximum coverage $\max_j C[j]$ allowing at most one MEM per node of $G$.}
Use \Cref{lemma:forward} to find all forward propagation links\;
\For{$k' \gets 1$ to $k$}{
  Initialize data structures $\mathcal{T}^a_{k'}$ and $\mathcal{T}^b_{k'}$ with keys $(x+\kappa'-1,j)$ such that $([x..x+\kappa'-1],(i,v,i+\kappa'-1))= A[j]$, $1 \leq j \leq N$, and with key $(0,0)$, all keys associated with values $-\infty$\;
  $\mathcal{T}^a_{k'}. \mathsf{update}((0,0),0)$\; 
  $\mathcal{T}^b_{k'}. \mathsf{update}((0,0),0)$\; 
}
\tcc{Save to $\mathtt{anchors}[v]$ all node MEMs of node $v$.}
\For{$j \gets 1$ to $N$}{
    $([x..x+\kappa'-1],(i,v,i+\kappa'-1))=A[j]$\;
	$\mathtt{anchors}[v].\mathsf{push}(j)$\;
    $C^-[j] \gets 0$\;
    $C [j] \gets \kappa'$\;
}

\For{$v \in V$ in topological order}{
    \For{$j \in \mathtt{anchors}[v]$}{
       \tcc{Update the data structures for every path that covers $v$, stored in $\mathtt{paths}[v]$.}
       $([x..x+\kappa'-1],(i,v,i+\kappa'-1))=A[j]$\;       
       \For{$k' \in \mathtt{paths}[v]$}{
          $\mathcal{T}^a_{k'}. \mathsf{upgrade}((x+\kappa'-1,j),C[j])$\;
          $\mathcal{T}^b_{k'}. \mathsf{upgrade}((x+\kappa'-1,j),C^-[j]-x)$\;
       }
    }
   \tcc{\textbf{PROPAGATE FORWARD STARTS}}
   \For{$(w,k') \in \forward[v]$}{ 
      \For{$j \in \mathtt{anchors}[w]$}{
           $([x..x+\kappa'-1], (i,v,i+\kappa'-1))=A[j]$\;       
           $C^\mathtt{a}[j] \gets \mathcal{T}^a_{k'}.\mathsf{RMaxQ}(0,x-1)$\;
           $C^\mathtt{b}[j] \gets x+\mathcal{T}^b_{k'}.\mathsf{RMaxQ}(x,x+\kappa'-1)$\;
           $C^-[j] \gets \max(C^-[j],C^\mathtt{a}[j],C^\mathtt{b}[j])$\;
           $C[j] = C^-[j]+\kappa'$\; 
       }
    }
   \tcc{\textbf{PROPAGATE FORWARD ENDS}}
}
\Return{$\argmax_j C[j]$}\;
\caption{\label{algo:colinearchainingDAGkpathcover}
Asymmetric co-linear chaining between a sequence and a DAG using a path cover and node MEMs.}
\end{algorithm}

Just like the original algorithm, our simplified version fills a table $C[1..N]$ such that $C[j]$ is the maximum coverage of an asymmetric chain that uses the $j$-th node MEM as its last item. That is, there is an asymmetric chain that induces a subsequence $Q'$ of the query $Q$ of length $C[j]$. In addition, our simplified version is restricted to chains that can include at most one MEM per node and includes an intermediate step to fill table $C^-[1..N]$ such that $C^-[j]=C[j]-\kappa'$, where $\kappa'$ is the length of the $j$-th node MEM. The reason for these modifications will become clear when we integrate the algorithm with the symmetric string-to-string chaining.

To fill tables $C[1..N]$ and $C^-[1..N]$, the algorithm considers a) MEMs from different nodes without overlap in the query and b) MEMs from different nodes with overlap in the query. These cases are illustrated in the left panel of \Cref{fig:cases}. The algorithm maintains the following data structure for each case and for each path in a given path cover of $k$ paths (see e.g.\ {\cite[Chapter 5]{de2000computational}}):

\begin{lemma}
\label{lemma:searchtree}
The following four operations can be supported with a balanced binary search tree $\mathcal{T}$ in time $O(\log n)$, where $n$ is the number of key-value pairs $((k,j), \mathtt{val})$ stored in the tree. Here $k$ is the primary key, $j$ is the secondary key to break ties, and $k,j,\mathtt{val}$ are integers.
\begin{itemize}
\item $\mathsf{value}(k,j)$: Return the value associated to key $(k,j)$ or $-\infty$ if $(k,j)$ is not a proper key. 
\item $\mathsf{update}((k,j),\mathtt{val})$: Associate value $\mathtt{val}$ to key $(k,j)$.
\item $\mathsf{upgrade}((k,j),\mathtt{val})$: Associate value $\max(\mathtt{val},\mathsf{value}(k,j))$ to key $(k,j)$.
\item $\mathsf{RMaxQ}(l,r)$: Return $\max_{l\leq k \leq r, (k,j) \text{ is a key in } \mathcal{T}} \mathtt{value}(k,j)$ (\emph{Range Maximum Query}), or $-\infty$ if the range is empty.
\end{itemize}
Moreover, the balanced binary search tree can be constructed in $O(n)$ time, given the $n$ pairs $((k,j),\mathtt{val})$ sorted by component $(k,j)$.  
\end{lemma}

The algorithm processes the nodes in topological order, keeping the invariant that once node $v$ is visited, the final values $C[j]$ and $C^-[j]$ are known for all anchors $j$ included in node $v$. These values are then stored in the search trees.
As a final step in the processing of $v$, the information stored in the search trees is propagated forward to nodes $w$, where $v$ is the last node reaching $w$ on some path-cover path. This propagated information is used for updating the intermediate values for MEMs at node $w$.
These forward links are preprocessed with the following lemma:

\begin{lemma}[Adaptation of {~\cite[Lemma 3.1]{Maketal19}}]
\label{lemma:forward}
Let $G = (V,E)$ be a DAG, and let $P_1,\dots,P_k$ be a path cover of $G$. We can compute in $O(k^2|V|)$ time the set of \emph{forward propagation links} $\forward[u]$ defined as follows: for any node $v$ and path $k'$, $(v,k') \in \forward[u]$ if and only if
 $u$ is the last node on path $k'$ that reaches $v$ such that $u \neq v$.
\end{lemma}
\begin{proof}
    The original DP algorithm~\cite{Maketal19} runs in $O(k|E|)$ time, but recently it has been shown~\cite[Algorithms 6 and 7]{kritikakis2022fast} how to do this in time $O(k|E_{red}|)$, where $E_{red}$ are the edges in the transitive reduction of $G$. Finally, C\'aceres et~al.~\cite{caceres2022minimum,Arieletal22} showed a transitive sparsification scheme proving that $|E_{red}| \le k|V|$.
\end{proof}

Data structures $\mathcal{T}^a_{k'}$ store as primary keys all ending positions of MEMs in $Q$ and as values the corresponding $C[j]$s for node MEMs $A[j]$ processed so far and reaching path $P_{k'}$ (line 15). When a new node MEM is added to a chain at line 20, the range query on $\mathcal{T}^a_{k'}$ guarantees that only chains ending before $v$ in $G$ and before the start of the new node MEM in $Q$ are taken into account. Data structures $\mathcal{T}^b_{k'}$ also store as primary keys all ending positions of node MEMs in $Q$, but as values they store the values $C^-[j]$ with an invariant subtracted (line 16). This invariant is explained by the range query at line 21, that considers chains overlapping (only) in $Q$ with the new node MEM to be added: consider the chain ending at node MEM $A[j']=([x'..x'+\kappa''-1],(i',v',i'+\kappa''-1))$ and the new node MEM $A[j]=([x..x+\kappa'-1],(i,v,i+\kappa'-1)$ is to be added to this chain, where $x\leq x'+\kappa''-1\leq x+\kappa'-1$. This addition increases the part of $Q$ covered by the chain (excluding the new node MEM) by $x-x'$. This is exactly the value computed at line 21, maximizing over such overlapping node MEMs.

\begin{figure}
    \centering
    \includegraphics[width=.7\textwidth]{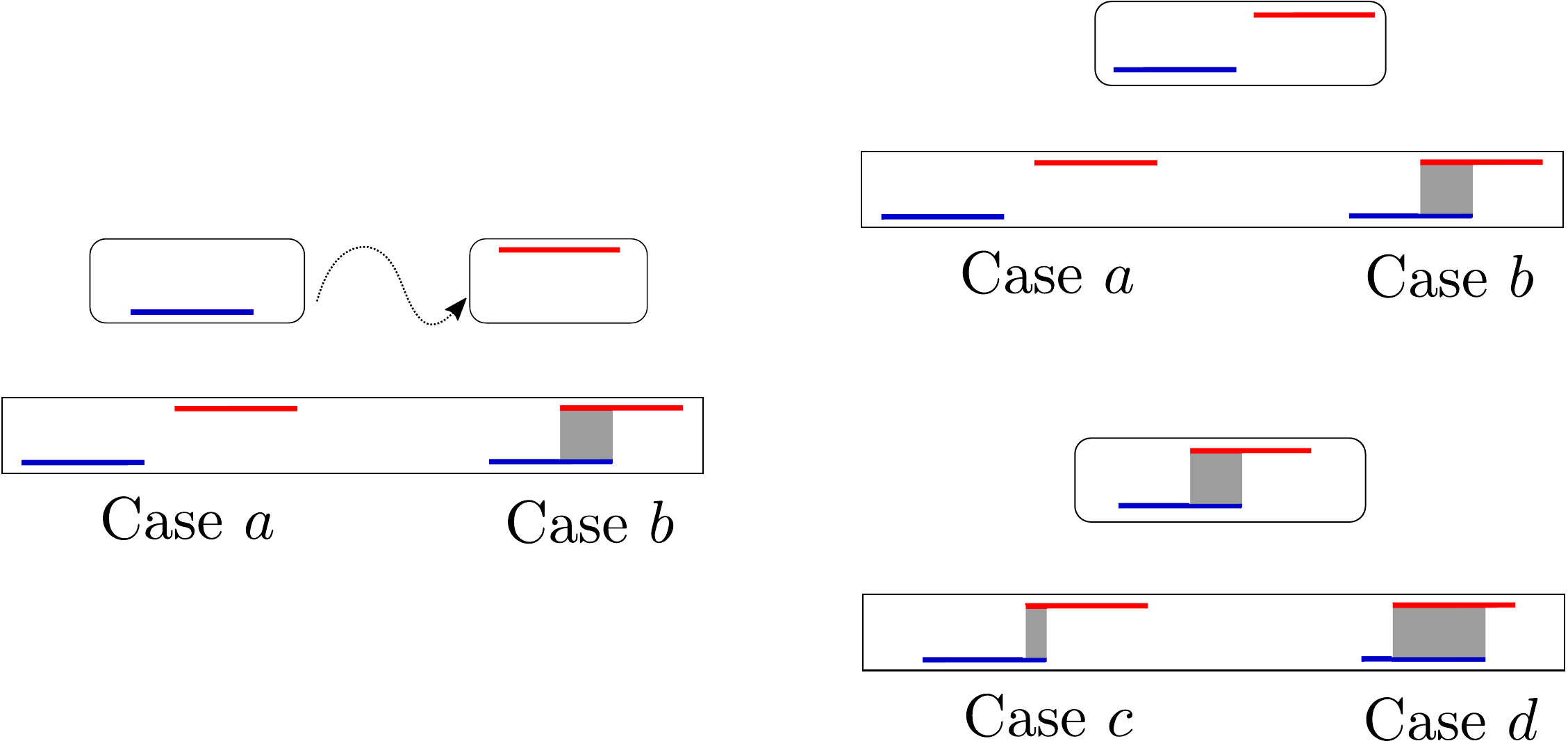}
    \caption{Precedence of MEMs partitioned to three classes (left, top right, and bottom right subfigures) by occurrence in graph/text (top part of each subfigure) and thereafter to two out of total four cases that require different data structure on the query (bottom part of each subfigure). }
    \label{fig:cases}
\end{figure}

\subsection{Revisiting symmetric string-to-string chaining with MEMs\label{sect:symmetricCLCrevisited}} 

Before modifying the algorithm to properly consider overlaps of node MEMs in $G$, let us first modify the symmetric string-to-string chaining algorithm of M\"akinen and Sahlin~\cite[Algorithm 2]{MS20} to harmonize the notation and to consider the simplification of \cite[Theorem 6]{MS20} that applies in the case of (string) MEMs.
This modification computes the optimal chain given MEMs $A[1..N]$ between strings $T$ and $Q$ and is given as \Cref{algo:colinearchainingTwoSidedOverlaps}.

\begin{algorithm}[H]
\KwIn{An array $A[1..N]$ of (string) MEMs $(x,i,\kappa')$ between $Q$ and $T$.}
\KwOut{Index of a MEM ending a chain with maximum coverage $\max_j C[j]$.}
Initialize data structures $\mathcal{T}^\mathtt{a}$ and $\mathcal{T}^\mathtt{b}$ with keys $(x+\kappa'-1,j)$ and data structures  $\mathcal{T}^\mathtt{c}$ and $\mathcal{T}^\mathtt{d}$ with keys $(x-i,j)$, where $(x,i,\kappa')=A[j]$, $1 \leq j \leq N$, and all trees with key $(0,0)$. Associate values $-\infty$\ to all keys.\\
$\mathcal{T}^\mathtt{a}.\mathsf{upgrade}((0,0),0)$\; 
$M=\{(x,j) \mid (x,i,\kappa')=A[j], 1\leq j \leq N\} \cup \{(x+\kappa'-1,j) \mid (x,i,\kappa')=A[j], 1\leq j \leq N\}$\; 
$M.sort()$\; 
\For{$(x',j) \in M$}{
  $(x,i,\kappa')=A[j]$\;
  \If{$x==x'$}{ \tcc{Start of MEM.}
      $C^\mathtt{a}[j] = \mathcal{T}^\mathtt{a}.\mathsf{RMaxQ}(0,x-1)$\;
      $C^\mathtt{b}[j] = x+ \mathcal{T}^\mathtt{b}.\mathsf{RMaxQ}(x,x+\kappa'-1)$\;
      \HiLi $C^\mathtt{c}[j] = i+ \mathcal{T}^\mathtt{c}.\mathsf{RMaxQ}(-\infty,x-i)$\;
      \HiLi  $C^\mathtt{d}[j] = x+ \mathcal{T}^\mathtt{d}.\mathsf{RMaxQ}(x-i+1,\infty)$\;      
      $C^-[j] = \max(C^\mathtt{a}[j],C^\mathtt{b}[j],C^\mathtt{c}[j],C^\mathtt{d}[j])$\;
      $C[j] = C^-[j]+\kappa'$\;
      \HiLi $\mathcal{T}^\mathtt{c}.\mathsf{upgrade}((x-i,j),C^-[j]-i)$\;
      \HiLi $\mathcal{T}^\mathtt{d}.\mathsf{upgrade}((x-i,j),C^-[j]-x)$\;
  }
  \Else{ \tcc{End of MEM.}
    $\mathcal{T}^\mathtt{a}.\mathsf{upgrade}((x+\kappa'-1,j),C[j])$\;
    $\mathcal{T}^\mathtt{b}.\mathsf{upgrade}((x+\kappa'-1,j),C^-[j]-x)$\;
    \HiLi $\mathcal{T}^\mathtt{c}.\mathsf{update}((x-i,j),-\infty)$\;
    \HiLi $\mathcal{T}^\mathtt{d}.\mathsf{update}((x-i,j),-\infty)$\;
  }
}
\Return{$\argmax_j C[j]$}\;
\caption{\label{algo:colinearchainingTwoSidedOverlaps}
Symmetric chaining with two-sided overlaps using MEMs.}
\end{algorithm}

The algorithm uses the same two data structures as before to handle the cases illustrated at the top right of \Cref{fig:cases}. Moreover, the two additional data structures (balanced binary search trees) in \Cref{algo:colinearchainingTwoSidedOverlaps} handle the overlaps in $T$ by dividing the computation further into cases c) and d) illustrated at the bottom right of \Cref{fig:cases}): c) if two MEMs overlap more in $T$ than in $Q$, tree $\mathcal{T}^\mathtt{c}$ is used for storing the solution; d) otherwise, tree $\mathcal{T}^\mathtt{d}$ is used for storing the solution.
We refer to the original work \cite{MS20} for the derivation of the invariants and the range queries to handle these cases. The handling of these cases is highlighted with gray background in \Cref{algo:colinearchainingTwoSidedOverlaps}. 

\subsection{Integration of symmetry to DAG chaining\label{sect:symmetricDAGCLCintegration}} 

We will now merge the two algorithms from previous subsections to solve Problem~\ref{prob:symmetricdagchaining}. This algorithm is shown as \Cref{algo:colinearchainingDAGTwoSidedOverlaps}; lines highlighted with a dark gray background are from \Cref{algo:colinearchainingTwoSidedOverlaps}, whereas lines highlighted with a light gray background are a hybrid of both, and the rest are from \Cref{algo:colinearchainingDAGkpathcover}. When visiting node $v$ the algorithm executes the steps of \Cref{algo:colinearchainingTwoSidedOverlaps} on anchors included in $v$, with $C^\mathtt{a}[j]$ and $C^\mathtt{b}[j]$ having already been updated with anchors not included in $v$ through forward propagation identical to \Cref{algo:colinearchainingDAGkpathcover}. The hybrid parts reflect the required changes to \Cref{algo:colinearchainingDAGkpathcover} in order to visit the MEM anchors twice as in \Cref{algo:colinearchainingTwoSidedOverlaps}. This merge covers all three cases of \Cref{fig:cases}. 

\begin{theorem}
Given labeled DAG $G=(V,E,\ell)$ with path cover $P_1$, \dots, $P_k$, query string $Q$, and set $A[1..N]$ of node MEMs between $Q$ and $G$,
\Cref{algo:colinearchainingDAGTwoSidedOverlaps} solves the symmetric DAG chaining with overlaps problem (\Cref{prob:symmetricdagchaining}) in time $O(k^2|V|+kN \log N)$.
\label{thm:DAGLCS}
\end{theorem}

\begin{corollary}\label{cor:lcs}
The length of a longest common subsequence (LCS) between a path in a labeled DAG $G=(V,E,\ell)$ and string $Q$ can be computed in time $O(n+m+k^2|V|+|E|+ kN\log N)$, where $m=|Q|$, $n$ is the total length of node labels, $k$ is the width (minimum number of paths covering the nodes) of $G$, and $N$ is the number of node MEMs. 
\end{corollary}
\begin{proof}
The node MEMs can be computed in time $O(n+m+N)$ with \Cref{lemma:nodeMEMs}. A minimum path cover with $k$ paths can be computed in $O(k^2|V|+|E|)$ time~\cite{Arieletal22,caceres2022minimum}. Forward propagation links can be computed in $O(k^2|V|)$ time with \Cref{lemma:forward}. Finally, the term $kN\log N$ comes from \Cref{thm:DAGLCS}. The connection between LCS and solution to symmetric chaining follows with identical arguments as in the proof of \Cref{obs:LCS} in \Cref{app:proof-lcs-string}. If $P$ is a path containing an LCS of length $c$, then \Cref{algo:colinearchainingDAGTwoSidedOverlaps} finds a chain of coverage exactly $c$ as its execution considers the corresponding chain between $\ell(P)$ and $Q$ as done in \Cref{algo:colinearchainingTwoSidedOverlaps}. In this case node MEMs are not necessarily MEMs between $\ell(P)$ and $Q$, but exact matches supporting the necessary character matches, see \Cref{app:proof-lcs-string}. 
\end{proof}
Note that the LCS connection can be easily adapted for long MEMs spanning two or more nodes of $G$: \Cref{def:mem} considers all (string) MEMs between $Q$ and $\ell(P)$, for any arbitrary path $P$; we did not consider symmetric chains of long MEMs due to the difficulty of handling path overlaps efficiently (see also \cite{Maketal19}).

\begin{algorithm}[H]
\KwIn{Same as in \Cref{algo:colinearchainingDAGkpathcover}.}
\KwOut{Index of a MEM ending at a chain with maximum coverage $\max_j C[j]$ allowing overlaps in $G$.}
Use \Cref{lemma:forward} to find all forward propagation links.\\
\For{$k' \gets 1$ to $k$}{
Initialize data structures $\mathcal{T}^\mathtt{a}_{k'}$ and $\mathcal{T}^\mathtt{b}_{k'}$ with keys $(x+\kappa'-1,j)$ and key $(0,0)$, and data structures  $\mathcal{T}^\mathtt{c}_{k'}$ and $\mathcal{T}^\mathtt{d}_{k'}$ with keys $(x-i,j)$, where $([x..x+\kappa'-1], (i,v,i+\kappa'-1))=A[j]$, $1 \leq j \leq N$. Associate values $-\infty$\ to all keys.\\
$\mathcal{T}^a_{k'}. \mathsf{update}((0,0),0)$\; 
$\mathcal{T}^b_{k'}. \mathsf{update}((0,0),0)$\; 
}
Initialize arrays: $\mathtt{anchors}$, $C^-$ and $C$ as in \Cref{algo:colinearchainingDAGkpathcover}\; 
\For{$v \in V$ in topological order}{
   \HiLiHybrid $M=\{(x,j) \mid ([x..x+\kappa'-1], (i,v,i+\kappa'-1))=A[j], j\in  \mathtt{anchors}[v] \} \cup \{(x+\kappa'-1,j) \mid ([x..x+\kappa'-1], (i,v,i+\kappa'-1))=A[j],  j\in  \mathtt{anchors}[v]\}$\; 
   \HiLiHybrid $M.sort()$\; 
   \tcc{Update the data structures for every path that covers $v$, stored in $\mathtt{paths}[v]$.}
   \For{$k' \in \mathtt{paths}[v]$}{
      \HiLiHybrid \For{$(x',j) \in M$}{
         \HiLiHybrid $(x,i,\kappa')=A[j]$\;
         \HiLi \If{$x==x'$}{ \tcc{Start of MEM.}
             \HiLi $C^\mathtt{a}[j] = \mathcal{T}^\mathtt{a}_{k'}.\mathsf{RMaxQ}(0,x-1)$\;
             \HiLi $C^\mathtt{b}[j] = x+ \mathcal{T}^\mathtt{b}_{k'}.\mathsf{RMaxQ}(x,x+\kappa'-1)$\;
             \HiLi $C^\mathtt{c}[j] = i+ \mathcal{T}^\mathtt{c}_{k'}.\mathsf{RMaxQ}(-\infty,x-i)$\;
             \HiLi $C^\mathtt{d}[j] = x+ \mathcal{T}^\mathtt{d}_{k'}.\mathsf{RMaxQ}(x-i+1,\infty)$\; 
             \HiLi$C^-[j] = \max(C^-[j],C^\mathtt{a}[j], C^\mathtt{b}[j], C^\mathtt{c}[j],C^\mathtt{d}[j])$\;
             \HiLi $C[j] = C^-[j]+\kappa'$\;
             \HiLi $\mathcal{T}^\mathtt{c}_{k'}.\mathsf{upgrade}((x-i,j),C^-[j]-i)$\;
             \HiLi $\mathcal{T}^\mathtt{d}_{k'}.\mathsf{upgrade}((x-i,j),C^-[j]-x)$\;
         }
         \HiLiHybrid \Else{ \tcc{End of MEM.}
             \HiLiHybrid $\mathcal{T}^\mathtt{a}_{k'}.\mathsf{upgrade}((x+\kappa'-1,j),C[j])$\;
             \HiLiHybrid $\mathcal{T}^\mathtt{b}_{k'}.\mathsf{upgrade}((x+\kappa'-1,j),C^-[j]-x)$\;
             \HiLi $\mathcal{T}^\mathtt{c}_{k'}.\mathsf{update}((x-i,j),-\infty)$\;
             \HiLi $\mathcal{T}^\mathtt{d}_{k'}.\mathsf{update}((x-i,j),-\infty)$\;
        }
      }    
    }
    Execute \textbf{PROPAGATE FORWARD} subroutine of \Cref{algo:colinearchainingDAGkpathcover}\;
}
\Return{$\argmax_j C[j]$}\;
\caption{\label{algo:colinearchainingDAGTwoSidedOverlaps}
Symmetric co-linear chaining between a sequence and a DAG using a path cover and node MEMs.}
\end{algorithm}

\No{
\section{Experiments}

The most prominent application for chaining in graphs is to map long erroneous reads on the graph. Our LCS formulation is not optimal for this application, as it does not penalize gaps: a good quality local alignment could be artificially extended with distant MEM matches at both ends. To alleviate this deficiency, there have been recent developments to extend the chaining framework to consider gap costs for traditional read mapping~\cite{JGT22} as well as for sequence to graph alignment~\cite{GC23}. However, for the latter case it appears difficult to make such formulation fully symmetric due to there being exponential many paths between two anchors. Due to these considerations, we decided to implement a practical read alignment framework that combines ideas from our work with those of the earlier work. 

Our proposal includes the following components:
\begin{itemize}
    \item Construct an elastic founder graph $G$ from a multiple sequence alignment (MSA)~\cite{RM22a},
    \item find node MEMs between $G$ and a set of reads (\cref{sect:nodeMEMs}),
    \item project the matches to MSA coordinates,
    \item use gap-sensitive chaining to find an alignment between each read to the MSA~\cite{CJ23},
    \item and project the alignments back to the graph. 
\end{itemize}

\section*{Acknowledgments} 

This project has received funding from the European Union’s Horizon 2020 research and
innovation programme under the Marie Skłodowska-Curie grant agreement No.\ 956229, and from the Academy of Finland grants No.\ 352821 and 328877.
}

\section{Discussion\label{sect:discussion}}

In this paper, we focused on MEMs with no lower threshold on their length to achieve the connection with LCS. In practical applications, chaining is sped up by using as anchors only MEMs that are of length at least $\kappa$, a given threshold. Just finding all such $\kappa$-MEMs is a non-trivial problem and solvable in sub-quadratic time only on some specific graph classes~\cite{RCM23b}. However, once such $\kappa$-MEMs are found, one can split them to node-MEMs and then apply Algorithm~\ref{algo:colinearchainingDAGTwoSidedOverlaps} to chain them. The resulting chain optimizes the length $|C|$ of a longest common subsequence $C$ between the query $Q$ and a path $P$ such that each match $C[k]=Q[i_k]=\ell(P)[j_k]$ is supported by an exact match of length at least $\kappa$, where $1\leq k \leq |C|$, $i_1<i_2<\cdots <i_{|C|}$, and $j_1<j_2<\cdots <j_{|C|}$. That is, there is a $\kappa$-MEM $([x_k,y_k],[c_k,d_k])$ with respect to $Q$ and $\ell(P)$ s.t. $x_k\leq i_k \leq y_k$ and $c_k\leq j_k \leq d_k$ for each $k$. Additionally, Ma et~al.~\cite[Appendix C]{Ma2022.01.07.475257} showed that asymmetric co-linear chaining can be extended to graphs with cycles by considering the graph of the strongly connected components. In the extended version of this paper we will show how to combine our results to obtain symmetric chaining in general graphs.
\bibliography{correctedbiblio}

\clearpage
\appendix

\section{Chaining for longer MEMs}\label{app:longMEMs}

We now show that graph MEMs of \Cref{def:mem} can be captured simply by concatenating node MEMs. Given two graph substrings $(i,P=u_1..u_k,j)$ and $(i',P'=v_1..v_{k'},j')$, they can be concatenated into $(i,P,j) \cdot (i',P',j')$ only if one of the following two conditions hold: $(u_k,v_1) \in E$, $j = \lVert  u_k \rVert $, and $i' = 1$; or $u_k = v_1$ and $i' = j+1$.
In the former case, $(i,P,j) \cdot (i',P',j') \coloneqq (i, P \cdot P', j')$, whereas in the latter case $(i,P,j) \cdot (i',P',j') \coloneqq (i, u_1..u_{k-1} \cdot v_1..v_{k'}, j')$. We then say that two MEMs $([x..y], (i,P,j))$ and $([x'..y'], (i',P',j'))$ can be concatenated if $(i,P,j)$ can be concatenated to $(i',P',j')$ and $x' = y+1$, and in such case we analogously define $([x..y], (i,P,j)) \cdot ([x'..y'], (i',P',j') \coloneqq ([x..y'], (i,P,j) \cdot (i',P',j')$.

Let $A$ be the set of node MEMs found using algorithm of \Cref{sect:nodeMEMs}. We call a sequence of node MEMs $A'[1..p] \subseteq A$ a \emph{perfect chain} if $A'[j]$ can be concatenated to $A'[j+1]$ for $1\leq j < p$. Note that the concatenation of all such node MEMs in the perfect chain yields a longer exact match.

\begin{customthm}{1}
For every MEM $([x..y], (i,P,j))$ between $G$ and $Q$, there is a perfect chain $A'[1..p]\subseteq A$ such that $A'[1] \cdots A'[p] = ([x..y], (i,P,j))$.
\end{customthm}
\begin{proof}
Let path $P$ be spanning nodes $v_1$, $v_2$, $\ldots$, $v_p$ and spelling $\ell(v_1)[i..\lVert v_1 \rVert]$ $\ell(v_2) \cdots \ell(v_{p-1})$ $\ell(v_p)[1..j]$. That is, there exist exact matches $([i_1..i_2-1],(i,$ $v_1,\lVert v_1\rVert ))$, $([i_2..i_3-1],(1,v_2,\lVert v_2\rVert ))$, $\ldots$, $([i_{p-1}..i_{p}-1]$, $(1,v_{p-1},\lVert v_{p-1}\rVert ))$, $([i_p..$ $i_{p+1}-1],$ $(1,v_p,j))$ between $Q$ and $G$. It is clear that if those matches are node MEMs then they form a perfect chain as they can be concatenated. Indeed, matches $([i_l..i_{l+1}-1], (1,v_l,\lVert v_l\rVert ))$ are right-maximal for $1 < l < p$ since they end at the end of a node label. For the same reason: matches $([i_l..i_{l+1}-1], (1,v_l,\lVert v_l\rVert ))$ are left-maximal for $1 < l < p$; $([i_1..i_2-1], (i,v_1,\lVert v_1\rVert ))$ is right-maximal; $([i_p..i_{p+1}-1],(1,v_p,j))$ is left-maximal. Finally, if we suppose by contradiction that match $([i_1..i_2-1], (i,v_1,\lVert v_1\rVert ))$ ($([i_p..i_{p+1}-1],(1,v_p,j))$) can be extended to the left (right) to $([i_1-1..i_2-1], (i-1,v_1,\lVert v_1\rVert ))$ ($([i_p..i_{p+1}],(1,v_p,j+1))$) we contradict the maximality of $([x..y], (i,P,j))$.
\end{proof}

\begin{customcrl}{1}
The set $A$ is a \emph{compact representation} of the set $M$ of MEMs between query $Q$ and a labeled DAG $G=(V,E,\ell)$: it holds $|A|\leq \lVert M\rVert $, where $\lVert M\rVert $ is the length of the encoding of the paths in MEMs as the explicit sequence of its nodes. 
\end{customcrl}
\begin{proof}
The corollary follows from \Cref{theorem:longMEMs} and the fact that for every node MEM using node $v$ there is at least one MEM between $Q$ and $G$ whose path contains $v$. Indeed, $v$ can be used in multiple MEM paths.
\end{proof}

\section{Co-linear chaining on strings using MEMs \\ gives LCS}\label{app:proof-lcs-string}

We first prove~\cite[Theorem 7]{MS20}\footnote{We provide this proof for completeness since the original proof is incomplete as checked with co-author M\"akinen.}. A string $C[1..\ell]$ is an LCS of strings $Q$ and $T$ if it is a longest string that can be written as $C = Q[y_1]..Q[y_\ell] = T[j_1]..T[j_\ell]$ with $1 \le y_1 < .. < y_\ell \le |Q|$ and $1 \le j_1 < .. < j_\ell \le |T|$. Given a set $A$ of anchors being exact matches between $Q$ and $T$, we define an \emph{anchor-restricted LCS} if it is a longest string such that it can be written as before but additionally for every character match $Q[y_l] = T[j_l]$ there exists an anchor $(x_l,i_l,\kappa'_l) \in A$ such that $x_l \le y_l \le x_l+\kappa'_l-1$, $i_l \le j_l \le i_l+\kappa'_l-1$ (the anchor supports the character match) and $y_l-x_l = j_l-i_l$ (the match occurs within the same offset in the anchor). 

\begin{theorem}[{\cite[Theorem 7]{MS20}}]\label{thm:linear-lcs}
Given a set of anchors $A$ of exact matches between two strings $Q$ and $T$, the length of an anchored-restricted LCS equals the coverage of a maximum coverage chain under the co-linear chaining formulation of M\"akinen and Sahlin~\cite{MS20}.
\end{theorem}
\begin{proof}
    The authors of~\cite{MS20} proved that every chain $A'[1..N']$ of anchors induces a common subsequence between $Q$ and $T$ whose length equals the coverage of $A'$: each anchor $A'[l] = (x_l,i_l,\kappa'_l)$ contributes $c_l$ characters to this subsequence such that $c_l$ is the minimum between the characters of $[x_l...x_l+\kappa'_l-1]$ not covered by the rest of the chain $A[l+1..N']$ and the characters of $[i_l...i_l+\kappa'_l-1]$ not covered by the rest of the chain $A[l+1..N']$. We now prove that if $c$ is the length of an anchored-restricted LCS, then there is a \emph{weak chain}~\cite{MS20} of $A$ with coverage $c$, where weak chain is such that consecutive anchors $(x_l, i_l, \kappa'_l), (x_{l+1}, i_{l+1}, \kappa'_{l+1})$ of a weak chain satisfy $x_l < x_{l+1}$ and $i_{l} < i_{l+1}$. Our proof technique consists in filtering out anchors supporting the LCS (while preserving the coverage of the chain) so that the final set of anchors corresponds to a weak chain.
    
    Let $C[1..c]$ be an anchored-restricted LCS such that $C = Q[y_1]..Q[y_c] = T[j_1]..T[j_c]$ with $1 \le y_1 < .. < y_c \le |Q|$ and $1 \le j_1 < .. < j_c \le |T|$, and let $(x_l,i_l,\kappa'_l)$ the anchor supporting the match $Q[y_l] = T[j_l]$ for $1\le l\le c$, that is $x_l \le y_l \le x_l+\kappa'_l-1$ and $i_l \le j_l \le i_l+\kappa'_l-1$. We will show that we can remove anchors from the beginning of the chain so that (after the removal) $x_1 < x_2$ and $i_1 < i_2$ (the first anchor \emph{weakly precedes} the second anchor) while maintaining the coverage. The proof follows inductively by removing the first anchor and applying the same procedure in the rest of the chain and the rest of the anchored-restricted LCS until no anchors remain. We first show that we can obtain a non-strict inequality $x_1 \le x_2$ and then how to filter anchors when $x_1 = x_2$. The argument for $i_1$ and $i_2$ follows symmetrically.
    
    Consider the case where $x_1 \ge x_2$, then $x_2 \le x_1 \le y_1 < y_2 \le x_2+\kappa'_2-1$ that is, the second anchor is also covering $Q[x_1]$. If $i_1 \ge i_2$, then $i_2 \le i_1 \le j_1 < j_2 \le i_2+\kappa'_2-1$, thus the second anchor is also supporting the first match $Q[i_1]$ and thus we can safely remove the first anchor without changing the coverage of the chain. Otherwise $i_1 < i_2$, and suppose that the second anchor does not cover $T[i_1]$ (if it does we can remove the first anchor as before), that is $j_1 < i_2$ (the case $j_1 > i_2+\kappa'_2-1$ does not exist since $j_1 < j_2 \le i_2 + \kappa'_2-1$). In this case we can replace the match $Q[y_1] = T[j_1]$ by the match $Q[x_2] = T[j_2]$, which is covered by the second anchor and thus we can safely remove the first anchor (as we have discovered another anchored-restricted LCS). Indeed, the character match $Q[x_2] = T[i_2]$ exists since the second anchor is an exact match and it can replace the match $Q[y_1] = T[j_1]$ since it does not interfere with $Q[y_2] = T[j_2]$ ($x_2 \le x_1 \le y_1 < y_2$, and $j_2 > i_2$ since $j_2-i_2 = y_2- x_2 >0$).
\end{proof}
\begin{corollary}
The chaining algorithm by M\"akinen and Sahlin \cite{MS20} computes the length of an LCS between strings $Q$ and $T$ if it is given all (string) MEMs between $Q$ and $T$ as input anchors.
\label{obs:LCS}
\end{corollary}
\begin{proof}
It suffices to note that every character match of an LCS of length $c$ is supported by some MEM within the same offset (as one can start the match there and extend it to the left and right character by character), and thus, by \Cref{thm:linear-lcs}, the chaining algorithm by M\"akinen and Sahlin~\cite{MS20} finds a chain of coverage at least $c$.
\end{proof}
\end{document}